\documentclass{article}
\usepackage{spconf,amsmath,graphicx,amssymb,amsthm}

\newtheorem{theorem}{Theorem}

\title{COMPRESSIVE SIGNAL PROCESSING WITH CIRCULANT SENSING MATRICES}
\name{Diego Valsesia \qquad Enrico Magli}
\address{Politecnico di Torino (Italy) -- Dipartimento di Elettronica e Telecomunicazioni \thanks{This work is supported by the European Research Council under the European Community’s Seventh Framework Programme (FP7/2007-2013) / ERC Grant agreement n.279848.}}
\begin{document}
%
\maketitle
\begin{abstract}
Compressive sensing achieves effective dimensionality reduction of signals, under a sparsity constraint, by means of a small number of random measurements acquired through a sensing matrix. In a signal processing system, the problem arises of processing the random projections directly, without first reconstructing the signal. In this paper, we show that circulant sensing matrices allow to perform a variety of classical signal processing tasks such as filtering, interpolation, registration, transforms, and so forth, directly in the compressed domain and in an exact fashion, \emph{i.e.}, without relying on estimators as proposed in the existing literature. The advantage of the techniques presented in this paper is to enable direct measurement-to-measurement transformations, without the need of costly recovery procedures.
\end{abstract}
\vspace*{-0.2cm}
\section{Introduction}
\label{sec:intro}
\vspace*{-0.2cm}
Compressive sensing (CS) \cite{donoho2006cs} has successfully shown that a small number of measurements, obtained through a suitable sensing matrix, is able to acquire a sparse signal and enable its exact recovery from its measurements. Remarkably, knowledge of the sparsity basis is not explicitly needed at the moment of acquisition, but only for recovery. Since the beginnings of CS, it was clear that working in the compressed domain presented advantages in terms of complexity because many tasks do not really require recovery but rather to solve inference problems. The seminal paper by Davenport \emph{et al.} \cite{CompressiveSP} on signal processing with compressive measurements provided techniques to perform detection, classification, estimation and interference cancellation tasks. In this paper, we focus on a different, although similar, problem. The goal is to perform classical signal processing operations such as filtering, interpolation and others directly on the compressive measurements, thereby avoiding to first apply the computationally expensive reconstruction algorithms in order to later apply the operator on the reconstructed samples. The resulting operation is a measurement-to-measurement transformation, which takes as input the acquired measurements and outputs the measurements of the processed signal, without ever performing reconstruction. This is achieved by using a circulant sensing matrix, or a distributed circulation property in single node and multi-node scenarios respectively. Conversely, [2] is concerned with parameter estimation from the measurements. This can also be applied to the estimation of transform or filtered coefficients, but the estimate is always affected by some error, whereas our proposed method yields an exact solution.

\section{BACKGROUND AND NOTATION}
\label{sec:bkg}
\vspace*{-0.25cm}
We use the subscript $A_{[a,b]}$ to denote a submatrix of $A$ composed by the rows indexed by interval $[a,b]$. The subscript $\mathbf{x}_{\rightarrow a}$ denotes a circular shift to the right (downwards for column vectors) by $a$ positions.

In the standard CS framework, introduced in \cite{candes2006nos}, a signal $\mathbf{x}\in\mathbb{R}^{n\times 1}$ which has a sparse representation in some basis $\Psi\in\mathbb{R}^{n\times n}$, \textit{i.e.},
$ \mathbf{x} = \Psi \boldsymbol{\theta},\quad \Vert \boldsymbol{\theta} \Vert_0 = k,\quad k\ll n $,
can be recovered by a smaller vector $\mathbf{y}\in\mathbb{R}^{m\times 1}$, $k<m<n$, of linear measurements $\mathbf{y} = \Phi\mathbf{x}$, where $\Phi\in\mathbb{R}^{m\times n}$ is the \emph{sensing matrix}. The optimum solution, seeking the sparsest vector compliant with $\mathbf{y}$, is an NP-hard problem, but one can resort to a convex optimization reconstruction problem by minimizing the $l_1$ norm, provided enough measurements ($m\sim k\log(n/k)$) are available.
This algorithm is robust when it is used to reconstruct signals which are not exactly sparse, but rather compressible, meaning that the magnitude of their sorted coefficients (in some basis $\Psi$) decays exponentially.

The most used sensing matrices in the literature are random matrices whose elements are i.i.d. Gaussian random variables. However, circulant matrices have been recently proposed \cite{BajwaToeplitz}\cite{rauhutcirculant} because they can be implemented easily or arise naturally in various problems. Circulant matrices have been shown to be as effective as Gaussian sensing matrices, when the signal is acquired in its sparsity domain and have a slightly reduced performance otherwise \cite{YinPracticalCirculant}. The form of such matrices is
\vspace*{-0.1cm}
\begin{align*}
\Phi =
\left[\begin{array}{ccccc}
\phi_{1} & \phi_{2} & \phi_{3} & \cdots & \phi_{n}\\
\phi_{n} & \phi_{1} & \phi_{2} & \cdots & \phi_{n-1}\\
 &  & \vdots\\
\phi_{n-m+2} & \phi_{n-m+3} & \phi_{n-m+4} & \cdots & \phi_{n-m+1}
\end{array}\right]
\vspace*{-0.3cm}
\end{align*} 
where the first row, called seed in the following, is drawn at random (\emph{e.g.}, i.i.d. Gaussian). Other constructions (\emph{e.g.}, randomly selecting $m$ rows from the full matrix) are reported in \cite{YinPracticalCirculant}, but we will not consider them in this paper.
The $n$-point circular convolution of two sequences $\mathbf{x} \in \mathbb{R}^{n}$ and $\mathbf{h} \in \mathbb{R}^{N_f}$ is $\mathbf{x}_f = \mathbf{h} \circledast \mathbf{x} = H\mathbf{x}$, where $H$ is an $n\times n$ circulant matrix. 

The commutator of two linear operators $A \in \mathbb{R}^{n \times n}$ and $B \in \mathbb{R}^{n \times n}$ is zero if and only if the two operators commute and is defined as
\begin{align}
\label{eq:commutator}
\left[ A,B \right] = AB - BA.
\end{align}

\vspace*{-0.6cm}
\section{Single-node compressive filtering}
\vspace*{-0.3cm}
\label{sec:singlenode}
\begin{figure}[t]
\begin{minipage}[b]{1.0\linewidth}
\centerline{\includegraphics[width=0.99\columnwidth]{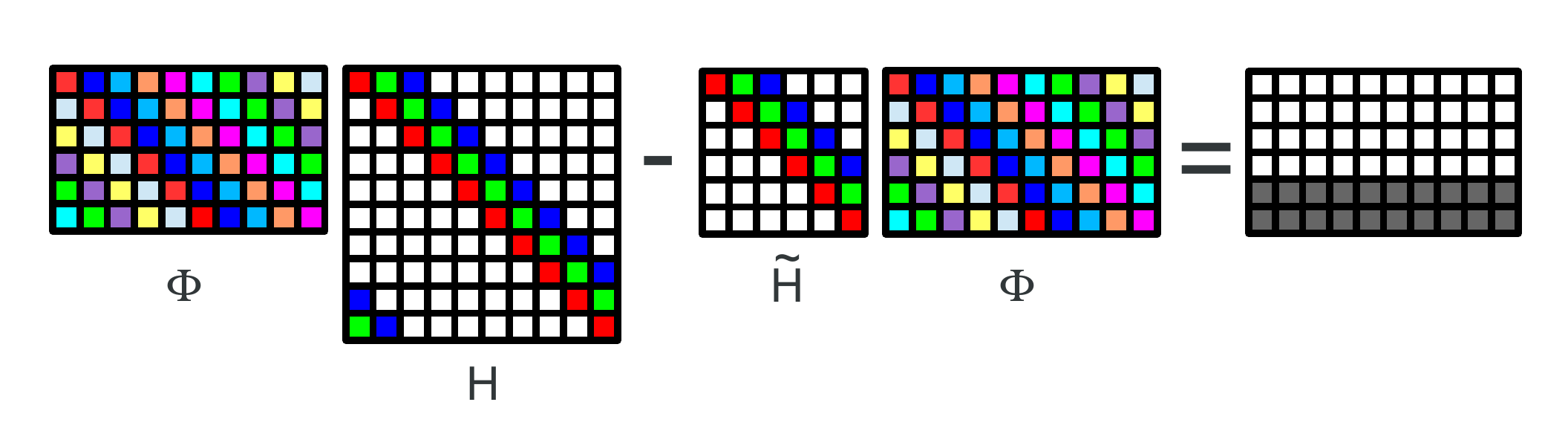}}
\vspace{-0.3cm}
\caption{Partial commutation condition}\medskip
\end{minipage}
\label{img:partial_comm}
\vspace*{-1.0cm}
\end{figure}

The main goal of this paper is to show that it is possible to compute a measurement-to-measurement transformation that allows us to find the measurements of a \emph{filtered} version of the acquired signal, directly from its linear measurements.
The classic filtering operation consists in convolving the filter impulse response with the signal, and circular convolution can be used if periodic boundary conditions are considered. This amounts to taking the product between a square circulant matrix $H$ and the signal of interest. It can be observed that square circulant matrices form a commutative group with respect to matrix product, hence if we used a square circulant sensing matrix $\Phi$ we would have $\mathbf{y}_f = \Phi H\mathbf{x} = H\Phi\mathbf{x} = H\mathbf{y}$. CS uses rectangular sensing matrices in order to achieve compression, hence $\Phi$ is a ``partial'' circulant matrix. Nevertheless, we can still exploit a partial commutation property. Let us define an extension of the commutator introduced in \eqref{eq:commutator} to handle rectangular matrices. The \emph{m-partial commutator} is defined as:
\begin{align}
\label{eq:partial_comm}
\left[ \Phi,H \right]_m = \Phi H - \tilde{H}\Phi 
\end{align}
where $\tilde{H}$ is the submatrix of $H$ restricted to the first $m$ rows and $m$ columns. The following theorem shows that obtaining the measurements of a filtered version of the signal is as straightforward as filtering the measurements. The price to pay for this operation is the ``corruption'' of $N_f-1$ measurements in the positions corresponding to the non-zero rows of $\left[ \Phi,H \right]_m$, thus the impulse response should not be too long to avoid corruption of too many measurements, and acquisition should take an extra $N_f-1$ measurements to account for this. A graphical depiction of Theorem \ref{thm:partial_commutator} is shown in Fig. 1. 

\begin{theorem}
\label{thm:partial_commutator}
Let $H$ be an $n\times n$ circulant matrix obtained from an impulse response of length $N_f$, $\Phi$ be an $m \times n$ partial circulant sensing matrix, measurements $\mathbf{y}=\Phi\mathbf{x}$, and measurements of the filtered signal $\mathbf{y}_f=\Phi H\mathbf{x}$. Then, $\left( \tilde{H}\mathbf{y} \right)_i = \left( \mathbf{y_f} \right)_i$ \ if and only if $i\in \left[ 1,m-N_f+1 \right]$.
\end{theorem}
\begin{proof}
First, we notice that it is enough to prove that only that the first $m-N_f+1 =m'$ rows of the $m$-partial commutator are zero. Let us call $C \in \mathbb{R}^{n\times n}$ the square circulant matrix having the same seed as $\Phi$, $H_T = H_{\left[ 1,m' \right]}$ and call $\tilde{H}_T = \tilde{H}_{\left[ 1,m' \right]}$. Then, $\left( \left[ C,H \right] \right)_{\left[ 1,m' \right]} = \left( CH \right)_{\left[ 1,m' \right]} - \left( HC \right)_{\left[ 1,m' \right]} = \left( \left[ \Phi,H \right]_m \right)_{\left[ 1,m' \right]} + \left( \tilde{H}\Phi \right)_{\left[ 1,m' \right]} - \left( HC \right)_{\left[ 1,m' \right]} =  \left( \left[ \Phi,H \right]_m \right)_{\left[ 1,m' \right]} +  \tilde{H}_T \Phi - H_T C  = 0$. Since, by assumption, the first row of matrix $H$ has at most $N_f$ non-zero entries (all in the first $N_f$ positions), we have $\left( H_T^T \right)_{\left[ m,n \right]}=0$ and $\tilde{H}_T = \left( \left( H_T^T \right)_{\left[ 1,m \right]} \right)^T$, so that $H_T C = \left( \left( H_T^T \right)_{\left[ 1,m \right]} \right)^T \Phi = \tilde{H}_T \Phi$. We thus find $\left( \left[ \Phi,H \right]_m \right)_{\left[ 1,m' \right]} = 0$. The last $N_f-1$ rows of $\left[ \Phi,H \right]_m$ are surely non-zero because $\left( H_T^T \right)_{\left[ m,n \right]}\neq 0$.
\end{proof}

\vspace*{-0.4cm}
\section{Multi-node compressive filtering}
\vspace*{-0.3cm}
\begin{figure}[t]
\begin{minipage}[b]{1.0\linewidth}
\centerline{\includegraphics[width=0.87\columnwidth]{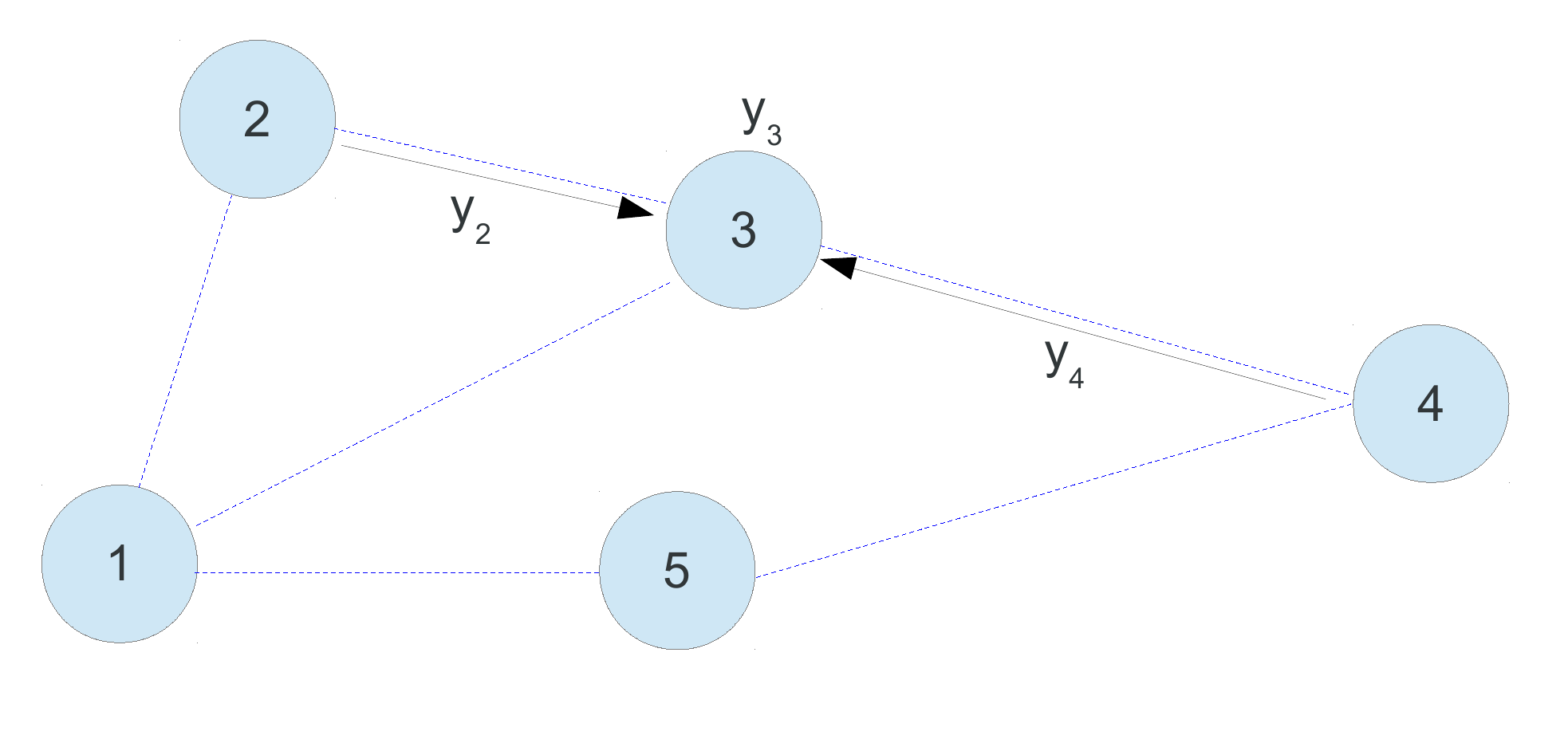}}
\vspace{-0.5cm}
\caption{Second derivative computation in multi-node scenario}\medskip
\end{minipage}
\label{distro_second}
\vspace*{-1.0cm}
\end{figure}

In this section we deal with an extension of the previous framework, which can encompass two scenarios. The first one considers a multi-node distributed system in which agents cooperate by sharing some measurements in order to perform the filtering operation. The second one is a single-node system having as many measurements as those from all the nodes in the previous interpretation, but with a structured sensing matrix. Let us describe the principles of operation referring to the former scenario for ease of explanation, but always keeping in mind the latter. 

In the multi-node scenario, each node is equipped with a random sensing matrix (\emph{e.g.}, Gaussian), \emph{not} circulant. However, nodes are assigned matrices that have a circulant property over the network. Numbering the nodes from 1 to $J$, this means that $\Phi^{(1)}$ is drawn at random, then row $i$ in matrix $\Phi^{(j)}$ is a circularly right-shifted version of the same row in $\Phi^{(j-1)}$. If all the nodes observe the same signal $\mathbf{x}$, then they can collaborate to perform the filtering operation in the compressed domain. Given an impulse response $\mathbf{h}$, node $j$ can obtain the measurements of the filtered signal $\mathbf{y}_f^{(j)} = \Phi^{(j)}\mathbf{x}_f$ as:
\vspace*{-0.2cm}
\begin{align}
\label{distro_filter}
\mathbf{y}_f^{(j)} = \sum_{i=0}^{N_f-1} h_i \mathbf{y}^{(j-i)} \qquad \text{for } j \in \left[ N_f , J \right]
\vspace*{-0.3cm}
\end{align}

We can derive a commutation theorem similar to Thm. \ref{thm:partial_commutator}, by defining the $(J,m)$-\emph{distributed partial commutator}:

\vspace*{-0.3cm}
\begin{align}
\label{eq:distributed_comm}
\left[ \tilde{\Phi},H \right]_{J,m} = \tilde{\Phi}H - \left( H_J \otimes I_m \right) \tilde{\Phi}
\vspace*{-0.2cm}
\end{align}
where $\tilde{\Phi} = \left[ \Phi^{(1)^T} \Phi^{(2)^T} \cdots \Phi^{(J)^T} \right]^T$, $H_J$ is a submatrix of $H$ restricted to the first $J$ rows and columns, $I_m$ is the $m \times m$ identity matrix, and $\otimes$ denotes Kronecker product.

\begin{theorem}
\label{thm:ditributed_partial_commutator}
Given $J$ nodes, the $l$-th measurement of the filtered signal at node $j$ is $\left( \mathbf{y}_f^{(j)} \right)_l = \sum_{i=0}^{N_f-1} h_i \left( \mathbf{y}^{(j-i)} \right)_l$ if and only if $j > N_f-1$. 
\end{theorem}

We skip the proof for brevity and because it follows that of Thm. \ref{thm:partial_commutator}, now using the $(J,m)$-distributed partial commutator.

The main advantage of this multi-node framework is that each node has a fully random matrix on its own. Therefore, if its number of measurements is sufficiently high, it can perform reconstruction without suffering the slight performance degradation shown by circulant matrices when the sparsity domain is not the identity. Moreover, no measurement is corrupted in this process, except for $N_f-1$ nodes, which cannot compute the filtered measurements with respect to their own sensing matrix.  Moreover, multiple nodes can stack their measurements, corresponding to the second scenario discussed at the beginning of the section. In this case, the resulting sensing matrix would be obtained by stacking the individual sensing matrices, thus having a block-circulant structure (a total of $m$ independent rows, used as seed of $J$ circulant blocks). 

\vspace*{-0.3cm}
\section{Main applications of the proposed framework}
\vspace*{-0.3cm}
The following sections describe some applications of the presented framework. For reasons of brevity, we shall focus on the single-node case, but all the results can be readily extended to the multi-node scenario. The general case of filtering using circular convolution with a sequence $\mathbf{h}$ has already been treated in the previous section. The following examples consider problems that can be cast in a similar manner.

\vspace*{-0.15cm}
\subsection{Finite differences}
\vspace*{-0.15cm}
It is of interest to apply a discretized differential operator to the signal and compute the measurements of the resulting sequence directly from the original measurements. Let us consider a discretization of the second derivative, so that the $n$-th sample in the output is computed, in vector form, as:
\begin{align*}
\vspace*{-0.2cm}
\mathbf{x}_{f} = \mathbf{x}_{\rightarrow 1} - 2\mathbf{x} + \mathbf{x}_{\leftarrow 1} 
\end{align*}
where the subscript $\rightarrow k$ denotes a circular right-shift by $k$ positions.
It is easy to check that the measurements $\mathbf{y}_{f} = \Phi\mathbf{x}_{f}$ can be obtained directly in the compressed domain as:
\begin{align*}
\mathbf{y}_{f} = \mathbf{y}_{\rightarrow 1} - 2\mathbf{y} + \mathbf{y}_{\leftarrow 1} 
\end{align*}
with the corruption of the first measurement.
In the multi-node setting it is immediate to derive that node $j$ can request the measurements of nodes $j-1$ and $j+1$ in order to compute $\Phi^{(j)}\mathbf{x}_f$ directly in the compressed domain:
\begin{align*}
\mathbf{y}_{f}^{(j)} = \mathbf{y}^{(j-1)} - 2\mathbf{y}^{(j)} + \mathbf{y}^{(j+1)} 
\end{align*}
This example is shown in Fig. 2.

\vspace*{-0.2cm}
\subsection{Compressive interpolation}
\vspace*{-0.2cm}
\label{cmp_interpolation}
The objective is to compute the measurements of an interpolated version of the signal of interest from its compressive measurements. Interpolation is modelled as the cascade of an upsampler and a low-pass filter. Therefore it is possible to perform the operation in the compressed domain if we assume that the original signal was acquired by a circulant matrix with a few zeroed columns, whose number depends on the upsampling factor. The low-pass filter is implemented by the standard compressive filtering technique explained in the previous sections, bearing in mind that $N_f-1$ measurements will be corrupted.
As an example, let us consider a piecewise linear interpolator with an interpolation factor of 2. It is easy to verify that the interpolated sequence and the corresponding measurement-domain transformation are:
\begin{align*}
\mathbf{x}_{\mathrm{INT}} = \frac{1}{2}\mathbf{x}_{\rightarrow 1} + \frac{1}{2}\mathbf{x}_{\leftarrow 1} + \mathbf{x} \\
\mathbf{y}_{\mathrm{INT}} = \frac{1}{2}\mathbf{y}_{\rightarrow 1} + \frac{1}{2}\mathbf{y}_{\leftarrow 1} + \mathbf{y}
\end{align*}
In this example the original measurements have been acquired as $\mathbf{y} = \hat{\Phi} \mathbf{x}$, where the $i$-th column of $\hat{\Phi}$ is $\hat{\Phi}_i = \Phi_{2i-1}$ for a given circulant matrix $\Phi$. The measurements of the interpolated signal with respect to matrix $\Phi$ are now available.

\vspace*{-0.1cm}
\subsection{Shift retrieval and registration}
\vspace*{-0.1cm}
The problem of integer shift retrieval in the compressed domain was considered in \cite{EldarShift}. The authors showed that the test $\max_s \mathfrak{R}\left\lbrace \langle \mathbf{z}, \Phi D^s \Phi^\star \mathbf{v} \rangle \right\rbrace$ exactly recovers the integer shift $s$ from as low as one measurement ($\mathbf{z}$ and $\mathbf{v}$ are the measurements of the signal and of its shifted version using cyclic shift operator $D^s$), given some conditions on $\Phi$. The partial Fourier matrix was considered as an example of sensing matrix satisfying the prescribed properties. We notice that circulant matrices do not satisfy the conditions listed in Theorem 1 in \cite{EldarShift}; in particular, $\alpha\Phi^\star\Phi \neq I$ for any scaling factor $\alpha$. However, we can exploit Theorem \ref{thm:partial_commutator} to perform shift retrieval in a different way. One of the implications of Theorem \ref{thm:partial_commutator} is that a shift in the signal corresponds to a shift in the measurements by the same amount, being the cyclic shift operator $D^s$ a circulant matrix, with the corruption of $s$ measurements.

\begin{theorem}
Let $\Phi \in \mathbb{R}^{m \times n}$ be a partial circulant matrix. Let $\mathbf{z}=\Phi\mathbf{x}$ and $\mathbf{v}=\Phi \mathbf{x}_{\rightarrow s^\star}$. If $\vert s^\star \vert < m$, then \eqref{csshiftretrieval} retrieves the correct shift $\hat{s}=s^\star$.
\begin{align}
\label{csshiftretrieval}
\hat{s} = \arg\min_s \Vert \tilde{\mathbf{z}} - \tilde{\mathbf{v}} \Vert_2
\end{align}
with $\tilde{\mathbf{z}} = \mathbf{z}_{[1,m-s]}$ and $\tilde{\mathbf{v}} = \mathbf{v}_{[s+1,m]}$ for $s\geq 0$, or $\tilde{\mathbf{z}} = \mathbf{z}_{[\vert s \vert+1,m]}$ and $\tilde{\mathbf{v}} = \mathbf{v}_{[1,m-\vert s \vert]}$ for $s<0$.
\end{theorem}
\begin{proof}
Let us call $C \in \mathbb{R}^{n\times n}$ the square circulant matrix having the same seed as $\Phi$. Suppose that $0 < s^\star < m$. If $s=s^\star$, we have $\tilde{\mathbf{v}}=C_{\left[ s^\star+1,m \right]}\mathbf{x}_{\rightarrow s^\star} = C_{\left[1,m-s^\star \right]}\mathbf{x}$, and $\tilde{\mathbf{z}}= C_{\left[ 1,m-s^\star \right]}\mathbf{x}$, hence $\tilde{\mathbf{v}}=\tilde{\mathbf{z}}$, so $\Vert \tilde{\mathbf{z}} - \tilde{\mathbf{v}} \Vert_2 =0$. If $s>s^\star$ or $s<s^\star$ we have $\tilde{\mathbf{v}}=C_{\left[ s+1,m \right]}\mathbf{x}_{\rightarrow s^\star} = C_{\left[s-s^\star+1,m \right]}\mathbf{x}$ and $\tilde{\mathbf{v}}=\left[ C_{\left[ m-s^\star,m \right]}^T C_{\left[ 1,m-s^\star-s \right]}^T  \right]^T \mathbf{x}$, respectively and $\tilde{\mathbf{z}}= C_{\left[ 1,m-s \right]}\mathbf{x}$ in both cases. Hence, by construction of $C$, we have $\tilde{\mathbf{v}}\neq\tilde{\mathbf{z}}$. The same reasoning applies for $-m < s^\star < 0$. Thus, we proved retrieval of the correct shift. When $\vert s^\star \vert \geq m$, $\tilde{\mathbf{v}}$ and $\tilde{\mathbf{z}}$ cannot be constructed because $\tilde{\mathbf{v}}$ and $\tilde{\mathbf{z}}$ contain measurements coming from two disjoint submatrices of C.
\end{proof}
\vspace*{-0.2cm}
Notice that thanks to compressive interpolation (section \ref{cmp_interpolation}), it is also possible to retrieve sub-integer shifts; we leave this as future work.

A different problem is registration. If we know the shift $s$, we can register a signal, meaning that we can compute the measurements of the shifted version simply by shifting the measurement vector. If want to compute the measurements with respect to the original matrix $\Phi$, then $s$ measurements will be corrupted by the registration. Otherwise, we can suppose that the registered measurements are measurements acquired with a different sensing matrix $\Phi'$, which is circulant, having as seed the $(m-s+1)$-th row of the full circulant matrix, that the partial $\Phi$ was extracted from.  

\vspace*{-0.1cm}
\subsection{Compressive wavelet transform}
\vspace*{-0.1cm}
We show that it is possible to obtain the measurements of the wavelet coefficients from the measurements of the signal directly in the compressed domain. Changing the signal domain after the sensing process can be useful in several ways. In \cite{ExploitingWaveletCS}, the authors propose a technique to improve the quality of the reconstruction of signals compressible in the wavelet domain by exploiting the tree-based structure within the wavelet coefficients. The method requires CS measurements of the wavelet coefficients in order to be applied. However, in many scenarios it is not practical to sense the wavelet coefficients, either because the sensing is implemented in a low-complexity, low-power hardware so that calculating the transform would be computationally and energetically costly or because specialized hardware directly acquires random projections of the data (\emph{e.g.}, \cite{RMPI}). Moreover, being able to compute the measurements of the wavelet coefficients, which typically hold a sparse or compressible representation of the signal, allows to recover them directly, thus avoiding any issues regarding the coherence \cite{CandesIncoherence} of sensing matrix and spasifying wavelet dictionary.

In order to implement the wavelet transform in the compressed domain we exploit the lifting scheme \cite{SweldensLifting}. The lifting method to compute the wavelet transform consists in a sequence of lifting steps, each composed by two filters operating on two subsequences. Initially, the two subsequences are the even and odd samples of the signal; the predict filter is used to predict the odd sequence from the even one, and the difference is then passed through the update filter and the result summed back to the even sequence. It can be noticed that this framework can be mapped in the compressed domain because the filtering steps can be implemented as explained in the previous sections under the choice of periodic boundary conditions. We thus only need to impose a structured sensing that acquires measurements of the even and odd subsequences separately. Let us call $\Phi^{(e)}$ and $\Phi^{(o)}$ the sensing matrices used to acquire the even and odd subsequences, respectively. Both have size $m\times n$, but the even (odd) columns are zeroed in $\Phi^{(o)}$ ($\Phi^{(e)}$). Moreover, the sum $\Phi^{(e)} + \Phi^{(o)}$ is a circulant matrix. The measurements of the even and odd sequences form the two input sequences in the compressive scheme, depicted in Fig. 3. The prediction step applies the prediction filter to the even sequence and the output is left-shifted by 1. The update step applies the update filter and right-shifts the output by 1. This is the structure of a single lifting step, which may be repeated depending on the particular transform to be implemented. Finally, gains are applied and the two sequences are summed. It can be shown that the final output is $\Phi\boldsymbol{\theta}$ being $\boldsymbol{\theta}$ the vector of wavelet coefficients.  We remark that, as explained in section \ref{sec:singlenode}, $2(N_f-1)+2(N_f-1)$ measurements per lifting step are corrupted. In more detail, let us consider a practical example taken from \cite{taubmanjpeg2000}. The spline 5/3 biorthogonal transform has the following prediction and update filters $\lambda_1(z)= -\frac{1}{2}(1+z)$ and $\lambda_2(z)= \frac{1}{4}(1+z^{-1})$ and gain factors $K_0=1$, $K_1=\frac{1}{2}$. Applying the compressive transform procedure we can obtain the measurements of the wavelet coefficients with a corruption of the first 2 and last 2 measurements, which will be discarded. 

\begin{figure}[t]
\begin{minipage}[b]{1.0\linewidth}
\centerline{\includegraphics[width=0.93\columnwidth]{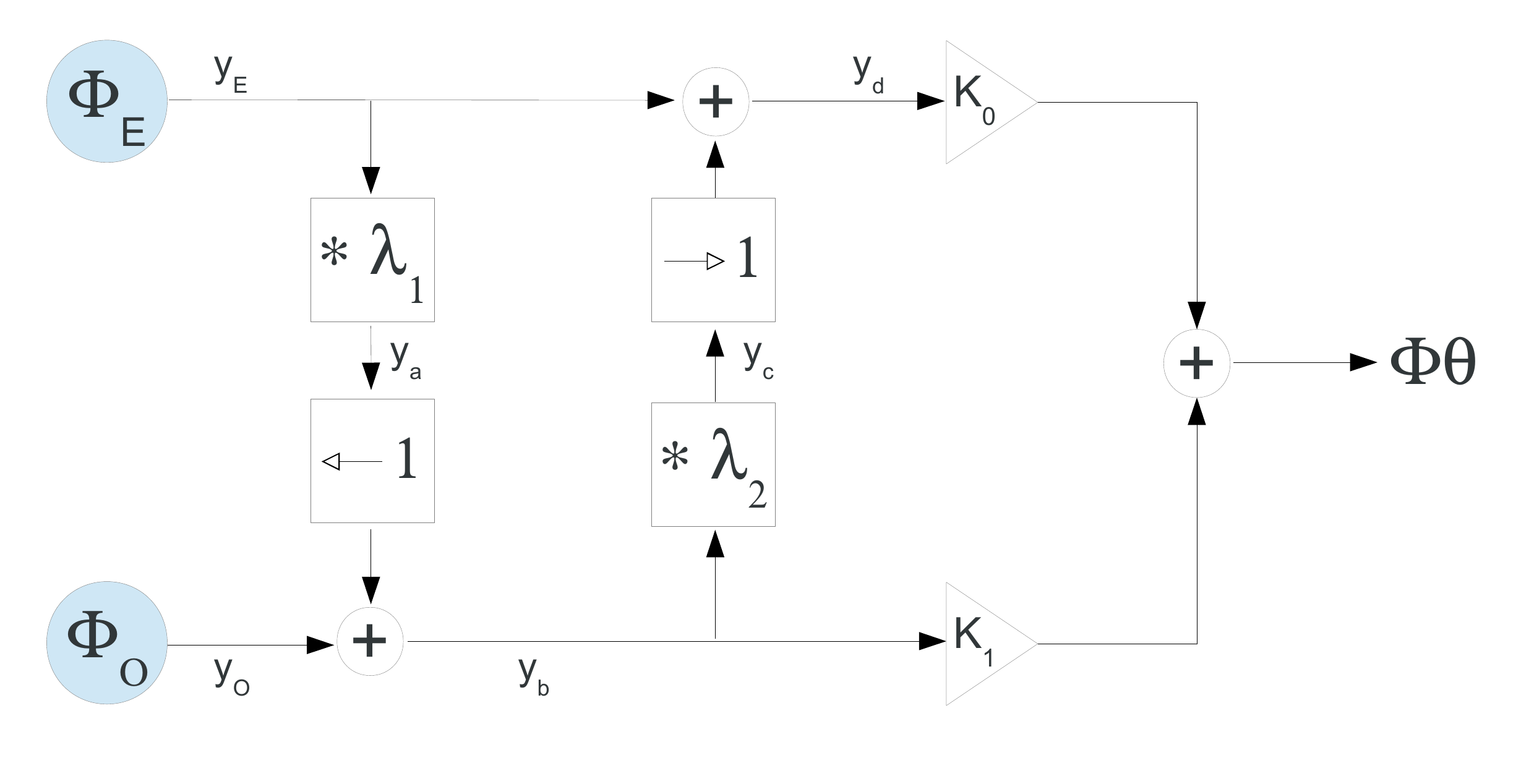}}
\vspace{-0.5cm}
\caption{Compressive wavelet transform}\medskip
\end{minipage}
\label{compressive_wavelet}
\vspace*{-1.0cm}
\end{figure}

\vspace*{-0.3cm}
\section{CONCLUSIONS}
\vspace*{-0.2cm}
In this paper we have shown how endowing the sensing matrix with circulant properties (\emph{e.g.}, fully circulant as in the single-node case or block-circulant as in the multi-node case) allows to translate a variety of classical signal processing tasks to the reduced dimensionality domain in an exact form or corrupting few measurements. We have discussed some applications of the presented paradigm, including various forms of filtering, shift retrieval and registration and a technique to transform the measurements of a signal into the measurements of its wavelet coefficients. Future directions for the work presented in this paper include extending the theory to 2D signals.

%



\end{document}